\documentclass[12pt]{article}

\usepackage[utf8]{inputenc}
\usepackage[T1]{fontenc}

\usepackage{setspace} 
\usepackage[vmargin = 1.25in, hmargin =1.25in]{geometry}

\usepackage[usenames,dvipsnames,svgnames]{xcolor}
\usepackage{colortbl}
\usepackage{amsmath, amssymb, amsfonts, graphicx, tikz,pdflscape, mathtools, amsthm, upgreek, bm,pgfplots,csvsimple,multirow, multicol, booktabs}
\usepackage{verbatim}
\usepackage{natbib}
\usepackage{accents}

\usepackage{needspace}

\usepackage{comment}
\usepackage[inline]{enumitem}
\usepackage{tocvsec2}
\usepackage{threeparttable}
\usetikzlibrary{calc, cd}
\setlist{noitemsep, topsep=0pt}
\usepackage[skip = 10pt]{caption}
\allowdisplaybreaks[1]
\allowdisplaybreaks[1]

\definecolor{Blue}{RGB}{86,180,233}
\definecolor{Orange}{RGB}{230,159,0}
\definecolor{Green}{RGB}{0,158,115}
\definecolor{GmailBlue}{RGB}{42, 93, 176} 
\usepackage[
	pagebackref,
	colorlinks=true,
	citecolor= GmailBlue,
	linkcolor=GmailBlue,
	urlcolor = GmailBlue
]{hyperref}

\newcommand{\bibtexorder}[1]{}

\usepackage{pgfplots}
\usepgfplotslibrary{groupplots,colorbrewer}
\pgfplotsset{compat=newest}
\pgfplotsset{cycle list/Set1}
\usepackage{tikz}
\usetikzlibrary{matrix,calc,shapes,arrows.meta,positioning}
\tikzset{
    vertex/.style = {shape=circle,draw, minimum size = 1.8em, inner sep = 0pt},
    edge/.style = {->,> = latex}
}

\usepackage[capitalize,noabbrev]{cleveref}

\newtheoremstyle{break}
{}
{}
{}
{}
{\bfseries}
{}
{\newline}
{}

\newtheorem*{theorem*}{Theorem}
\newtheorem*{cor*}{Corollary}

\newtheorem{lem}{Lemma}

\crefname{prop}{Proposition}{Propositions}
\crefname{thm}{Theorem}{Theorems}
\crefname{lem}{Lemma}{Lemmas}
\crefname{blem}{Lemma}{Lemmas}

\theoremstyle{definition}

\newtheorem*{rem*}{Remark}
\newtheorem*{claim*}{Claim}

\def\s{\sigma}

\def\D{\Delta}

\def\DD{\mathcal{D}}

\def\P{\mathbf{P}}

\DeclareMathOperator{\E}{\mathbf{E}}
\DeclareMathOperator{\supp}{supp} 

\DeclareMathOperator{\marg}{marg} 

\newcommand{\Brac}[1]{\left[ #1 \right]}

\title{Comment on Jackson and Sonnenschein (2007) ``Overcoming Incentive Constraints by Linking Decisions''}

\author{%
	Ian Ball%
	\thanks{Department of Economics, MIT, ianball@mit.edu.}
	\and 
	Matthew O. Jackson\thanks{Department of Economics, Stanford University,  jacksonm@stanford.edu}
	\and
	Deniz Kattwinkel%
	\thanks{Department of Economics, UCL, d.kattwinkel@ucl.ac.uk.}
}

\date{\today}

\begin{document}

\maketitle

We correct a bound in the definition of approximate truthfulness used in the body of the paper of \cite{jackson2007overcoming}.  The proof of their main theorem uses a different permutation-based definition, implicitly claiming that the permutation-version implies the bound-based version.  We show that this claim holds only if the bound is loosened. The new bound is still strong enough to guarantee that the fraction of lies vanishes as the number of problems grows, so the theorem is correct as stated once the bound is loosened. 

\paragraph{Setting}

Recall the setting of \citet[hereafter JS]{jackson2007overcoming}. Consider an $n$-agent collective decision problem $\DD = (D, U,P)$, where $D$ is the finite set of decisions; $U = U_1 \times \cdots \times U_n$ is the finite set of possible profiles of utility functions on $D$; and $P = (P_1, \ldots, P_n)$ in $\D(U_1) \times \cdots \times \D (U_n)$ is the profile of priors. 

There are $K$ independent copies of this decision problem, labeled $k =1 , \ldots, K$. Each agent $i$ knows their preference vector $u_i = (u_i^1, \ldots, u_i^K)$ in $U_i^K$, and their total payoff from a decision vector $(d^1, \ldots, d^K)$ is the sum $u_i^1(d^1) + \cdots + u_i^K (d^K)$. Utility functions are drawn independently across agents and decision problems, according to the priors in $P$. 

Given an ex ante Pareto efficient social choice function $f \colon U \to \D (D)$, JS introduce the following \textit{linking mechanism}. Each agent 
$i$ is asked to report a preference vector $\hat{u}_i = (\hat{u}_i^1, \ldots, \hat{u}_i^K)$ with exactly the same utility frequencies as the distribution $P_i^K$, where $P_i^K$ is the closest approximation of $P_i$ with the property that every probability is a multiple of $1/K$. In each decision problem, the mechanism applies the social choice function $f$ to the profile of preferences reported on that problem. 

To formalize the linking mechanism, define the \emph{marginal (distribution)} of a vector $u_i$ in $U_i^K$, denoted either $\marg u_i$ or $\marg (\cdot | u_i)$, by
\[
    \marg ( v_i | u_i) = \# \{ k: u_i^k = v_i\}/K, \qquad v_i \in U_i. 
\]
The linking mechanism is a pair $(M^K, g^K)$. The message space $M^K$ equals the product $M_1^K \times \cdots \times M_n^K$, where  
\[
    M_i^K = \{ \hat{u}_i \in U_i^K : \marg \hat{u}_i = P_i^K \}. 
\]
The outcome rule $g^K \colon M^K \to (\D(D))^K \subset \D (D^K)$ is defined by\footnote{%
This exact outcome rule is used only in the case $n =1$. For $n > 1$, the reports are modified before applying $f$ in such a way that the modified reports follow $P$ exactly.}
\[
    g^K ( \hat{u}^1, \ldots, \hat{u}^K) = (f( \hat{u}^1), \ldots, f(\hat{u}^K)),
\]  
where $\hat{u}^k$ denotes the vector  $(\hat{u}_1^k, \ldots, \hat{u}_n^k)$ of reports on problem $k$.

In the linking mechanism, an agent cannot report exactly truthfully if their realized preference vector violates the marginal constraint. JS focus on strategies in which each agent lies in as few decision problems as is feasible under the mechanism. They define a strategy $\s_i^K \colon U_i^K \to M_i^K$ to be \emph{approximately truthful} if 
\begin{equation} \label{eq:approx_truthful}
    \# \{k: [\s_i^K(u_i)]^k \neq u_i^k \} = \min_{\hat{u}_i \in M_i^K}  \# \{k: \hat{u}_i^k \neq u_i^k \},
\end{equation}
for all $u_i$ in $U_i^K$, where $[\s_i^K(u_i)]^k$ denotes the $k$-th component of $\s_i^K(u_i)$. The right side of \eqref{eq:approx_truthful} equals $K d( \marg u_i, P_i^K)$, where $d$ is the total variation metric on $\D (U_i)$ defined by $d (Q, Q') = \sum_{v_i \in U_i} ( Q(v_i) - Q'(v_i))_+$.\footnote{The right side of \eqref{eq:approx_truthful} is the minimal cost in the optimal transport problem between measures $K \marg u_i$ and $K P_i^K$ with cost function $c(x,y) = [x \neq y]$.}

We propose a weaker bound. A strategy $\s_i^K \colon U_i^K \to M_i^K$ is \emph{approximately truthful*} if 
\begin{equation} \label{eq:approx_truthful_star}
    \# \{k: [\s_i^K(u_i)]^k \neq u_i^k \} \leq \left(\# U_i-1 \right) K d( \marg u_i, P_i^K),
\end{equation} 
for all $u_i$ in $U_i^K$. The inequality in \eqref{eq:approx_truthful_star} relaxes \eqref{eq:approx_truthful} by a factor of $\left(\# U_i -1 \right)$. The agent can lie $\left(\# U_i -1 \right)$ times more than is required by the marginal constraint. The two definitions coincide if $\#U_i \leq 2$. For $\# U_i > 2$, the new definition is strictly weaker.  Both definitions extend immediately to mixed strategies.\footnote{A mixed strategy is approximately truthful (approximately truthful*) if it can be expressed as a mixture over approximately truthful  (approximately truthful*) pure strategies.} 

We also give a name to a different notion of truthfulness that appears in JS's proof. A strategy $\s_i^K \colon U_i^K \to \D (M_i^K)$ is \emph{permutation-truthful} if for each $u_i$ in $U_i^K$ and each $\hat{u}_i$ in $\supp \s_i^K (u_i)$ the following holds: for any subset $S$ of $\{1, \ldots, K\}$, if there is a 
bijection $\pi$ on $S$ such that
$\hat{u}_i^k = u_i^{\pi(k)}$ for all $k$ in $S$, then $\hat{u}_i^k = u_i^k$ for all $k$ in $S$. That is, the agent never nontrivially permutes their true preferences over a subset of decision problems. 

The concepts of permutation- and approximate truthfulness serve distinct roles in JS's argument. Permutation-truthfulness captures the reporting incentives created by the efficiency of the outcome function. Approximate truthfulness directly quantifies how frequently the reports match the truth.   

Theorem 1.i in JS (p.~248) says that each linking mechanism has a Bayesian equilibrium in approximately truthful strategies; however, their proof shows only that there exists a Bayesian equilibrium in (label-free\footnote{This means that whenever the agent's preference vector $u_i$ is permuted, their report $\s_i^K(u_i)$ is permuted in the same way; see JS (p.~251).}) permutation-truthful strategies. We give a counterexample to the existence of an approximately truthful equilibrium. Next we prove that permutation-truthful strategies are approximately truthful* and that this weaker property is still sufficient for part (ii) of Theorem 1. The remaining parts of the theorem are true as stated---only part (iii) mentions approximate truthfulness and it is true under either definition. Therefore, Theorem 1 is correct if \emph{approximately truthful} is everywhere replaced by \emph{approximately truthful*}.

\paragraph{Counterexample} \label{sec:countterexample}

Suppose there is a single agent ($n = 1$). The set of decisions is $D = \{ a, b, c\}$. The agent has three possible utility functions, denoted $u( \cdot | A)$, $u( \cdot | B)$, and $u( \cdot | C)$. The prior $P$ puts probability $1/3$ on each utility function. We say that the agent's \emph{type} is either $A$, $B$, or $C$. Suppose that type $A$ (respectively $B$, $C$) strictly prefers decision $a$ (respectively $b$, $c$) to the other two decisions. Therefore, the unique ex ante Pareto efficient social choice function is $f(A) = a$, $f(B) = b$, and $f(C) = c$. 

Consider linking $K = 3$ decisions. In the linking mechanism, the agent must report a vector $\hat{u} = (\hat{u}^1, \hat{u}^2, \hat{u}^3)$ in which $A$, $B$, and $C$ each appear exactly once. Suppose the agent has type vector $(A,A,B)$, as indicated in \cref{table:ce}. Reporting truthfully would violate the quota. Under an approximately truthful strategy, the agent must lie exactly once by reporting either $(A,C,B)$ or $(C,A,B)$. But the agent strictly prefers to lie twice by reporting $(A,B,C)$ if 
\begin{align*} \label{counter-exmp}
    u (b | A) + u(c | B) > u ( c | A) + u(b | B),
\end{align*}
which holds as long as $u(c|A)$ is low enough. In this case, reporting $C$ in problem $2$ is so costly that the agent prefers to report $B$ even though this forces them to lie again in problem $3$ to satisfy the quota.\footnote{In general, telling a different lie in one problem can start a cycle of up to $\#U_i-1$ lies in total. If the cycle had length $\#U_i$, then the agent could report truthfully on each problem in the cycle.} Reporting $(A,B,C)$ does not violate approximate truthfulness* or permutation-truthfulness.

\begin{table}
\begin{center}
\begin{tabular}{lc>{\columncolor[gray]{0.9}}c>{\columncolor[gray]{0.9}}cc}
    &  1& 2&3\\
    \hline
     type vector  & $A$ & $A$ & $B$ \\
     approximate truth  & $A$ & $C$ & $B$   \\
     deviation  & $A$ & $B$ & $C$ 
\end{tabular}
\end{center}
\caption{Counterexample}
\label{table:ce}
\end{table}

\paragraph{Approximate efficiency} \label{sec:efficiency}

Theorem 1.ii in JS says that their sequence $\{\s^K\}$ of equilibria \emph{approximate} $f$ in the sense that
\begin{equation} \label{eq:efficient}
 \lim_{K} \Brac{ \max_{k \leq K} \, \P \Big\{ g_k^K (\s^{K} (u)) \neq f(u^k) \Big\} } = 0,
\end{equation}
where $g_k^K$ denotes the component of $g^K$ in the $k$-th decision problem; the equation inside the probability is between lotteries in $\D(D)$; and the probability is taken over the random vector $u$ in $U^K$ and possible mixing in $\s^K$. By the definition of $g^K$, \eqref{eq:efficient} holds as long as we have
\begin{equation} \label{eq:truthful}
    \lim_{K} \Brac{ \max_{k \leq K} \, \P \bigl( [\s^{K} (u)]^{k} \neq u^k \bigr) } = 0.
\end{equation}

JS observe that \eqref{eq:truthful} follows from the law of large numbers for label-free approximately truthful strategy profiles. We confirm that \eqref{eq:truthful} holds for label-free approximately truthful* strategy profiles $\{\s^K\}$. For such profiles, we have for each agent $i$ that
\begin{equation*}
\begin{aligned}
    \max_{k \leq K} \, \P \bigl( [\s_i^{K} (u_i)]^{k} \neq u_i^k \bigr) 
    &\leq (\# U_i - 1) \E [d( \marg u_i, P_i^K)] \\
    &\leq  (\# U_i - 1) \E [d( \marg u_i, P_i)] + (\# U_i - 1) d ( P_i, P_i^K).
\end{aligned}
\end{equation*}
As $K \to \infty$, the approximation $P_i^K$ converges to $P_i$, and the expectation goes to zero by Glivenko--Cantelli (since $\#U_i$ is finite). Then \eqref{eq:truthful} follows by applying a union bound over the agents $i = 1, \ldots, n$.\footnote{For any (not necessarily label-free) approximately truthful* strategies $\s_i^K$, we can still conclude that the expected fraction of reports that are lies,  $\E[  \# \{ k: [\s_i^K (u_i)]^k \neq u_i^k \}]/K$, tends to $0$ as $K \to \infty$. This is the only property of approximately truthful strategies that is used in JS's proof of Theorem 1.iii-v.}

To complete the proof, we check that permutation-truthfulness implies approximate truthfulness*. Consider an agent $i$. The following combinatorial lemma guarantees that there is a sufficiently large subset $S$ of decision problems over which the agent permutes their true preferences. Under a permutation-truthful strategy, the agent is truthful on $S$, so the size of the complement of $S$ gives the desired bound \eqref{eq:approx_truthful_star} on the number of lies. 

\begin{lem} \label{combi-lem} Fix an agent $i$. For any pair of vectors $u_i$ and $\hat{u}_i$ in $U_i^K$, there exists a subset $S$ of $\{1, \ldots, K\}$ and a bijection $\pi$ on $S$ 
such that
\begin{enumerate}[label = (\roman*), ref = \roman*]
    \item  \label{it:perm} $\hat{u}_i^k = u_i^{\pi(k)}$ for all $k$ in $S$;
    \item \label{it:card} $\# S \geq K - (\# U_i -1) K d( \marg u_i, \marg \hat{u}_i)$.
\end{enumerate}
\end{lem}

\begin{proof} Given $u_i$ and $\hat{u}_i$ in $U_i$, construct a directed multigraph as follows. The set of nodes is $U_i$. For each $k = 1, \ldots, K$, add edge $k$ from node $u_i^k$ to node $\hat{u}_i^k$. In this graph, each node $v_i$ has out-degree $\deg^+(v_i) = K \marg (v_i | u_i)$ and in-degree $\deg^-(v_i) = K \marg (v_i | \hat{u}_i)$. 

Now we add \textit{new} edges as follows. Add an edge from a node with net out-degree to a node with net in-degree, update the degrees of the new graph, and repeat until the graph is balanced, i.e., $\deg^+ ( v_i) = \deg^-(v_i)$ for all $v_i$ in $U_i$. Let $K'$ be the number of new edges added by the this procedure. We have
\begin{equation*}
\begin{aligned}
   K' &=\sum_{v_i} (\deg^+ (v_i) - \deg^-(v_i))_+ \\
    &= K \sum_{v_i} ( \marg(v_i | u_i) - \marg (v_i | \hat{u}_i))_+ \\
    &=  K d (\marg u_i, \marg \hat{u}_i). 
\end{aligned}
\end{equation*}

Now we have a balanced graph with $K + K'$ edges. Partition this graph into edge-disjoint cycles.\footnote{To do so, start at a node with an outgoing edge. Form a path by arbitrarily selecting outgoing edges until the path contains a cycle. Remove the cycle and repeat. Since the graph remains balanced, this process can terminate only when every edge has been removed.} Remove every cycle that contains at least one of the new edges. Define $S$ to be the set of labels of the remaining edges. Since at most $K' \#U_i$ edges were removed, we have
\[
    \#S \geq K + K' - K' \#U_i = K -  (\#U_i - 1) K d (\marg u_i, \marg \hat{u}_i),
\]
so $S$ satisfies \eqref{it:card}. For \eqref{it:perm}, define $\pi$ on $S$ by letting $\pi(k)$ be the label of the edge that follows edge $k$ in its cycle. (In particular, $\pi(k) = k$ if edge $k$ is a loop.) Since the head of edge $k$ equals the tail of edge $\pi(k)$, we have $\hat{u}_i^k = u_i^{\pi(k)}$ by the definition of the graph. 
\end{proof}

\bibliographystyle{ecta}
\bibliography{lit.bib}

\end{document}